\documentclass[aps,pra,twocolumn,amsmath,amssymb,showpacs,floatfix]{revtex4-1}

\usepackage{times}

\usepackage[utf8x]{inputenc}
\usepackage[english]{babel}
\usepackage[T1]{fontenc}
\usepackage{lmodern}
\usepackage{amssymb}
\usepackage{bbold}
\usepackage{amsmath}
\usepackage{amsthm}
\usepackage[pdftex]{graphicx}
\usepackage{epstopdf}

\usepackage{dcolumn}
\usepackage{bm}
\usepackage{color}
\usepackage{relsize,dsfont,mathrsfs,empheq,verbatim,upgreek}
\usepackage{etoolbox}
\usepackage[caption = false]{subfig}
\usepackage{pgfplots}
\newtheorem{mydef}{Definition}
\newtheorem{mylem}{Lemma}
\newtheorem{mythe}{Theorem}

\providecommand{\norm}[1]{\lVert#1\rVert}

\begin{document}

\title{Generalized trace distance measure connecting quantum and classical non-Markovianity}

\author{Steffen Wi{\ss}mann and Heinz-Peter Breuer}
\affiliation{Physikalisches Institut, Universit\"at Freiburg, 
Hermann-Herder-Stra{\ss}e 3, D-79104 Freiburg, Germany}

\author{Bassano Vacchini}
\affiliation{Dipartimento di Fisica, Universit\`a degli Studi di Milano,
Via Celoria 16, I-20133 Milan, Italy}
\affiliation{INFN, Sezione di Milano, Via Celoria 16, I-20133 Milan, Italy}

 \begin{abstract}
  We establish a direct connection of quantum Markovianity of an open quantum system to its classical counterpart by generalizing the criterion based on the information flow. Here, the flow is characterized by the time evolution of Helstrom matrices, given by the weighted difference of statistical operators, under the action of the quantum dynamical evolution. It turns out that the introduced criterion is equivalent to P-divisibility of a quantum process, namely divisibility in terms of positive maps, which provides a direct connection to classical Markovian stochastic processes. Moreover, it is shown that similar mathematical representations as those found for the original trace distance based measure hold true for the associated, generalized measure for quantum non-Markovianity. That is, we prove orthogonality of optimal states showing a maximal information backflow and establish a local and universal representation of the measure. We illustrate some properties of the generalized criterion by means of examples.
  \end{abstract}
 
%Open systems, QInfo, ProbTheo, Quantum Optics removed
\date{\today}
\maketitle

\section{Introduction}
\label{sec:introduction}
Even though experimental techniques are developing rapidly, perfect isolation of fragile quantum systems from noisy environments cannot be warranted in general, which necessitates effective descriptions for nonunitary dynamics \cite{Breuer2007}. A well-known treatment of the dynamics of such systems is provided by a quantum dynamical semigroup represented by a generator of the Gorini-Kossakowski-Sudarshan-Lindblad form \cite{Lindblad,GKS}. 

However, this description uses several rather drastic approximations which do not apply to open quantum systems in general. The mismatch is typically associated to the neglect of memory effects. Classically, there exists a well-established mathematical theory dealing with stochastic dynamics featuring memory effects, based on the theory of stochastic processes, yet the definition of classical Markovian stochastic processes cannot be straightforwardly carried over to the quantum regime. This led to an intense debate along with different proposal for the characterization and quantification of memory effects in the dynamics of open quantum systems. Among others, approaches to non-Markovianity were based on deviation from semigroup dynamics \cite{Wolf}, on divisibility in terms of completely positive maps \cite{RHP}, on dynamics of entanglement \cite{RHP} and correlations \cite{Luo}, and on the Fisher information \cite{Sun}, see \cite{RHPreview} for a review.

In this work we focus on a straightforward extension of the measure of non-Markovianity introduced in Ref. \cite{BLP,BLP2} and reviewed in \cite{HPB2012}, which characterizes memory effects by an exchange of information between the open quantum system and the environment. This interpretation applies also to the generalized measure that uses Helstrom matrices known from quantum estimation theory \cite{Helstrom}. A similar approach has already been introduced in Ref. \cite{ChruscinskiRivas} focusing, at variance with the present proposal, on divisibility in terms of completely positive maps (CP-divisibility). However, CP-divisibility yields only a sufficient condition for quantum Markovianity with respect to the novel characterization proposed in this article, which turns out to be equivalent to divisibility in terms of positive maps (P-divisibility), and has in addition a clear-cut connection to classical Markov stochastic processes. The generalized criterion introduced in Section \ref{sec:information-flow-p} thus combines a physical interpretation of non-Markovianity in terms of an information flow and a relation to the well-known classical definition.

Moreover, the associated measure has similar mathematical features and representations as those found for the original definition \cite{LocalRep,OptimalStates} which simplify its analytical, numerical and experimental determination drastically. Indeed we prove in Section \ref{sec:expr-gener-non} that, first, optimal initial states for non-Markovian dynamics, experiencing a maximal backflow of information, must be orthogonal and, second, the measure admits a local representation showing locality and universality of quantum memory effects. 

An illustration of the essential features of the generalized characterization of non-Markovianity in comparison with the original definition and approaches based on CP-divisibility is provided in Section \ref{sec:examples}. The examples show the sensitivity of the novel approach to memory effects of dynamics to which the original definition was unsusceptible and illustrate the existing difference between different types of divisibility of a dynamical process. Finally, we summarize our results in Section \ref{sec:conclusions-outlook}.

\section{State discrimination and P-divisibility}
\label{sec:information-flow-p}

Henceforth, $\mathcal{H}$ refers to the Hilbert space of the open quantum system and the corresponding set of physical states, i.e., the set of positive trace class operators with unit trace, is denoted by $\mathcal{S}(\mathcal{H})$. Moreover, we assume that the dynamics of the open quantum system is determined by a one-parameter family $\Phi=\{\Phi_t|0\leq t\leq T\}$ of completely positive, trace preserving linear maps $\Phi_t$, where $\Phi_0$ refers to the identity \cite{Breuer2007}.

\subsection{Ensemble discrimination}
\label{sec:ensemble-discr}
The previously introduced definition of quantum non-Markovianity \cite{BLP,BLP2}, which was based on the concept of an information flow between system and environment, relied on the trace distance between two quantum states $\rho_1$ and $\rho_2$ defined by \cite{Nielsen}
\begin{equation}\label{eq:TraceDist}
\mathcal{D}(\rho_1,\rho_2)\equiv\tfrac{1}{2}\norm{\rho_1-\rho_2}_1~.
\end{equation}
Here $\norm{A}_1=\mathrm{Tr}|A|$, where the modulus of an operator is given by $|A|=\sqrt{A^\dag A}$, refers to the so-called trace norm on the set of trace class operators \cite{Hayashi}. A natural and interesting generalization of this quantity is obtained when one also allows for biased probability distributions. That is, one considers 
\begin{equation}\label{eq:TnormHelstrom}
 \norm{p_1\rho_1-p_2\rho_2}_1=\mathrm{Tr}|p_1\rho_1-p_2\rho_2|~,
\end{equation}
where $\{p_i\}$ refers to an arbitrary binary probability distribution, i.e. $p_{1,2}\geq0$ and $p_1+p_2=1$. Clearly, choosing an unbiased distribution, i.e. $p_{1,2}=1/2$, one obtains the original expression \eqref{eq:TraceDist}. This generalization was first proposed in this context by Chru\'sci\'nski \emph{et al.} \cite{ChruscinskiRivas}. 

The Hermitian operator $\Delta= p_1\rho_1-p_2\rho_2$ is also known as Helstrom matrix \cite{Helstrom,Holevo} and we find, applying the triangle inequality to Eq.\,\eqref{eq:TnormHelstrom},
\begin{equation}\label{eq:boundHelstrom}
 |p_1-p_2|\leq\norm{\Delta}_1\leq p_1+p_2=1~.
\end{equation}
Clearly, the lower bound is obtained for $\rho_1=\rho_2$ whereas the upper bound is attained if and only if $\rho_1\perp\rho_2$ meaning that the two states have orthogonal support (which is defined as the subspace spanned by the eigenvectors with nonzero eigenvalues). This is easily shown by means of the characterization of the trace norm
\begin{align}\label{eq:TnormHelstrom2}
 \norm{\Delta}_1&=2\max_\Pi \mathrm{Tr}\{\Pi\Delta\}+p_2-p_1~,
\end{align}
where the maximum is taken over all projection operators $\Pi$. Here, Eq.\,\eqref{eq:TnormHelstrom2} is derived employing the Jordan-Hahn decomposition of the Hermitian Helstrom matrix $\Delta$ in terms of two positive and orthogonal operators, i.e. $\Delta=S-Q$ where $S,Q\geq0$ and $S\perp Q$.

The interpretation of the trace distance of two states $\rho_{1,2}\in\mathcal{S}(\mathcal{H})$ as a measure for their distinguishability directly carries over to the trace norm of a Helstrom matrix $\Delta$. Consider a one-shot, two state discrimination problem where Alice prepares one out of two quantum states $\rho_{1,2}$ with corresponding probability $p_{1,2}$, i.e., we have $p_1+p_2=1$. Hence, we allow for the general situation of a biased preparation by Alice. She finally sends the prepared state to Bob who performs a single (strong) measurement to infer which state he had received \cite{Hayashi}. To guess the state from the measurement with possible outcomes $\Omega$, Bob defines two sets of possible results, $R$ and $\Omega\backslash R$, and assigns the state to be $\rho_1$ if the measurement outcome is in $R$ and $\rho_2$ if a value in $\Omega\backslash R$ is obtained. This strategy results in an effective two-valued positive operator valued measure (POVM) $\{T_R,\mathbb{1}-T_R\}$ where $T_R$ refers to the collection of effects corresponding to outcomes in $R$. The probability for correct state discrimination via this strategy is then given by
\begin{align}\label{eq:SuccProb}
 P_\mathrm{success}&=p_1\mathrm{Tr}\{T_R \rho_1\}+ p_2\mathrm{Tr}\{(\mathbb{1}-T_R) \rho_2\}\nonumber\\
&=p_2+\mathrm{Tr}\{\Delta T_R\}
\end{align}
which is maximal if $T_R$ is the projection $\Pi_{\{\Delta\geq0\}}$ on the subspace spanned by eigenvectors corresponding to positive eigenvalues of $\Delta$. Employing that $\mathrm{Tr}|X|=\mathrm{Tr}\{X(\Pi_{\{X\geq0\}}-\Pi_{\{X<0\}})\}$ for any Hermitian operator $X$ one shows that the maximal success probability for correct discrimination obeys \cite{Hayashi}

\begin{align}\label{eq:maxSuccProb}
 P_\mathrm{success}^\mathrm{max}=\max_{0\leq T_R\leq\mathbb{1}} P_\mathrm{success}=\frac{1}{2}\left(1+\norm{\Delta}_1\right)~.
\end{align}
Hence, the trace norm of the Helstrom matrix $\Delta=p_1\rho_1-p_2\rho_2$ is the bias in favor of the correct state identification of the state prepared by Alice so that it may be interpreted as a measure for the distinguishability of the two states $\rho_1$ and $\rho_2$ with associated weights $p_1$ and $p_2$.

Due to a result by Kossakowski \cite{Kossakowski,Kossakowski2} the trace norm can also be used to witness positivity of a trace preserving map $\Lambda$. That is, a trace preserving map $\Lambda$ is positive if and only if it defines a contraction for any Hermitian operator $X$ with respect to the trace norm, i.e.,
\begin{equation}\label{eq:contraction}
 \norm{\Lambda X}_1\leq\norm{X}_1~,~~\text{for all}~X=X^\dag~.
\end{equation}
We thus obtain for two states evolving according to a dynamical map $\Phi_t$
\begin{equation}\label{eq:contraction2}
 \norm{p_1\Phi_t(\rho_1)-p_2\Phi_t(\rho_2)}_1\leq\norm{p_1\rho_1-p_2\rho_2}~,
\end{equation}
for all $t$ where we also used linearity of the map. 

Adopting the previous characterization for quantum non-Markovianity \cite{BLP,BLP2} one then directly derives a generalized criterion which still relies on the concept of an information flow:
\begin{mydef}\label{def:Markov}
 A quantum process $\Phi$ is defined to be Markovian if $\norm{p_1\Phi_t(\rho_1)-p_2\Phi_t(\rho_2)}_1$ is a monotonically decreasing function of $t\geq0$ for all sets $\{p_i,\rho_i\}$ with $p_i\geq0$, $p_1+p_2=1$ and $\rho_i\in\mathcal{S}(\mathcal{H})$.
\end{mydef}

We stress that the authors of Ref. \cite{ChruscinskiRivas} introduced a different definition for Markovianity, which used the Helstrom matrices of a dilated system making it equivalent to CP-divisibility of the dynamical process. However, as we shall see in Section \ref{sec:dist-divis}, Markovianity as defined here can be related to the concept of P-divisibility and provides a clear-cut connection to classical Markov stochastic processes.

We conclude this section by showing that the previously developed explanation to substantiate the interpretation of quantum memory as an information backflow from the environment to the system \cite{ReviewHP} also applies here. To this end, we define the quantities
\begin{align}\label{eq:intInfo}
 \mathcal{I}_\mathrm{int}(t)&=\norm{p_1\rho_S^{(1)}(t)-p_2\rho_S^{(2)}(t)}_1~, \\
\intertext{and}
 \mathcal{I}_\mathrm{ext}(t)&=\norm{p_1\rho_{SE}^{(1)}(t)-p_2\rho_{SE}^{(2)}(t)}_1-\mathcal{I}_\mathrm{int}(t)~,
\label{eq:extInfo}
\end{align}
where $\rho_{SE}^{(i)}(t)=U_t\rho_S^{(i)}\otimes\rho_E U_t^\dag$ refers to the state of system and environment at time $t$ subject to a unitary evolution $U_t$. We thus have $\rho_S^{(i)}(t)=\mathrm{Tr}_E\rho_{SE}^{(i)}(t)$ and, similarly, the state of the environment at time $t$ is given by $\rho_E^{(i)}(t)=\mathrm{Tr}_S\rho_{SE}^{(i)}(t)$.

As discussed in \cite{ReviewHP} $\mathcal{I}_\mathrm{int}(t)$ describes the distinguishability of the open system at time $t$, while $\mathcal{I}_\mathrm{ext}(t)$ can be interpreted as the information on the total system which is not accessible by measurements on the open system only. Due to contractivity of the trace norm under positive and trace preserving maps (cf. Eq.\,\eqref{eq:contraction}) we conclude that both quantities are positive as the partial trace is (completely) positive and preserves the trace. Moreover, existence of the quantum dynamical map implies that the initial state is factorized, therefore $\mathcal{I}_\mathrm{ext}(0)=0$, so that unitary invariance of the trace norm implies
\begin{equation}
 \mathcal{I}_\mathrm{int}(t)+\mathcal{I}_\mathrm{ext}(t)=\mathcal{I}_\mathrm{int}(0)=\mathrm{const}~.
\end{equation}
This relation clearly expresses the idea of exchange of information between the open system and the environment as an increase of $\mathcal{I}_\mathrm{ext}(t)$ necessarily forces $\mathcal{I}_\mathrm{int}(t)$ to decrease. One may also derive upper bounds for the external information which, however, do not allow for a clear interpretation in terms of 
formation of correlations between system and environment and changes in the environmental states as proven for the trace distance in Ref. \cite{EPL-Corr}.

\subsection{Distinguishability and divisibility}
\label{sec:dist-divis}
In this section we introduce the concept of divisibility of a dynamical map which has been at the basis of several other proposals for quantum non-Markovian dynamics \cite{RHP,ChruscinskiRivas} and elucidate its connection to the distinguishability measure defined before.

We start by revisiting the notion of $n$-positivity. A map $\Lambda$ on $\mathcal{H}$ is said to be $n$-positive if and only if the linear tensor extension of the map given by $(\Lambda\otimes\mathbb{1}_n)(A\otimes B)=\Lambda(A)\otimes B$ for linear operators $A$ on $\mathcal{H}$ and $B$ on $\mathbb{C}^n$ is positive, that is, $(\Lambda\otimes\mathbb{1}_n)C\geq0$ for all positive linear operators $C$ on $\mathcal{H}\otimes\mathbb{C}^n$.

Clearly, $1$-positivity is equivalent to the map $\Lambda$ being positive and the concept of $n$-positivity is hierarchical, which means that $n$-positivity of $\Lambda$ implies that the map is also $k$-positive for all $1\leq k\leq n$ \cite{Blackadar}. The converse, however, is not true in general \cite{Choi,Stinespring}. Finally, one speaks about complete positivity if a map is $n$-positive for all $n\in\mathbb{N}$. For a finite-dimensional system with $\mathrm{dim}\mathcal{H}=N$ this property is equivalent to $N$-positivity \cite{Choi2}.

If the maps $\Phi_t$ comprising a dynamical process are invertible for all $t\geq0$ with inverse $\Phi_t^{-1}$, then one can define a two-parameter family of maps given by
\begin{equation}
 \Phi_{t,s}=\Phi_t\Phi_s^{-1}~,
\end{equation}
for all $t\geq s\geq0$ such that $\Phi_{t,0}=\Phi_t$ and $\Phi_t=\Phi_{t,s}\Phi_s$. We remark that although the existence of a left-inverse requiring injectivity of the maps $\Phi_t$ would be sufficient, we assume $\Phi_t^{-1}$ to be the left- and right-inverse. The notion of divisibility deals with the properties concerning $n$-positivity of these maps, i.e., 
\begin{mydef}\label{def:P-divisible}
A dynamical process $\Phi$ is called P-divisible (CP-divisible) if $\Phi_{t,s}$ is a positive (completely positive) map for all $t\geq s\geq0$.
\end{mydef}
Even though $\Phi_t$ is completely positive for all times by definition, this does not imply (complete) positivity of $\Phi_{t,s}$ in general as the inverse of a CP-map need not be positive showing that the above definition is nontrivial. However, the concept of divisibility relies on the existence of the two-parameter family which is limited to processes for which $\Phi_t^{-1}$ exists for all times. This property though cannot be taken for granted as e.g. the damped Jaynes-Cummings model on resonance \cite{Breuer2007,BLP2,ReviewHP} or examples on quantum semi-Markov processes \cite{SemiMarkov} show, thus making the concept of divisibility sometimes ill-defined. 

Typically, however, the inverse maps exist apart from isolated points in time so that one may describe the dynamical process by a time-local master equation for the open system during intermediate intervals. Assuming a sufficiently smooth time dependence, the generator obeys $\mathcal{K}_t=\dot{\Phi}_t\Phi_t^{-1}$ and has the general structure
\begin{align}\label{eq:MQ}
\mathcal{K}_t\rho_S(t)=&-i\left[H_S(t),\rho_S(t)\right]+\sum_j \gamma_j(t)\Bigl(A_j(t)\rho_S(t)A_j^\dag(t)\nonumber\\
&-\frac{1}{2}\bigl\{A_j^\dag(t)A_j(t),\rho_S(t)\bigr\}\Bigr)~,
\end{align}
which is similar to the well-known Lindblad form \cite{GKS,Lindblad} apart from the fact that the system Hamiltonian $H_S(t)$, the rates $\gamma_j(t)$ and the Lindblad operators $A_j(t)$ may depend on time as the process might not represent a semigroup. 

The maps of the two-parameter family are then given by
\begin{equation}
 \Phi_{t,s}=\mathrm{T}_\leftarrow \exp\left[\int_s^t\mathrm{d}\tau \mathcal{K}_\tau\right]~,
\end{equation}
where $\mathrm{T}_\leftarrow$ refers to chronological time-ordering. We thus have $\mathcal{K}_s=\frac{\mathrm{d}}{\mathrm{d}t}\Phi_{t,s}\mid_{t=s}$ so that we conclude: $\Phi$ is (C)P-divisible if and only if $\mathcal{K}_t$ is the generator of a (completely) positive semigroup for all fixed $t\geq0$. 

The famous Gorini-Kossakowski-Sudarshan-Lindblad theorem \cite{GKS,Lindblad} and a result on generators of positive semigroups by Kossakowski \cite{Kossakowski} then allow for a characterization of CP- and P-divisibility in terms of the generator \cite{ReviewHP}:

\begin{mythe}\label{the:Divisibility}
The dynamics generated by $\mathcal{K}_t$ \eqref{eq:MQ}
 \begin{enumerate}
\item  is CP-divisible if and only if $\gamma_j(t)\geq0$ holds for all $j$ and $t\geq0$. 
\item  is P-divisible if and only if for all $n\neq m$
\begin{equation}\label{eq:CondPdiv}
 \sum_{j}\gamma_j(t)|\langle m|A_j(t)|n\rangle|^2\geq0~,
\end{equation}
holds for any orthonormal basis $\{|n\rangle\}$ of $\mathcal{H}$ and all $t\geq0$. 
 \end{enumerate}
\end{mythe}
\begin{proof}
 The first statement is precisely the Gorini-Kossakowski-Sudarshan-Lindblad theorem \cite{GKS,Lindblad} and the second can be derived from Kossakowski's result on generators of positive semigroups \cite{Kossakowski} which states: A dynamics generated by $\mathcal{L}$ is P-divisible if and only if for any set of projections $\Pi=\{\Pi_m\}_{m\in I}$ defining a resolution of the identity, i.e. $\sum_m \Pi_m=\mathbb{1}_\mathcal{H}$, the following relations
\begin{align}
a_{mm}(\Pi)&\leq0~,~~ m\in I~,\label{eq:Cond1}\\
a_{mn}(\Pi)&\geq0~,~~ m\neq n\in I~,\label{eq:Cond2}\\
\sum_{m\in I}a_{mn}(\Pi)&=0~,~~ n\in I~.\label{eq:Cond3}
\end{align}
for $a_{mn}(\Pi)\equiv\mathrm{Tr}\{\Pi_m(\mathcal{L}\Pi_n)\}$ are satisfied.

It is easily seen that condition \eqref{eq:Cond3} refers to the preservation of trace which is always met for the generator $\mathcal{K}_t$ \eqref{eq:MQ} by its very structure. Rearranging terms of Eq.\,\eqref{eq:Cond3} one then finds
\begin{equation}
 a_{nn}(\Pi)=-\sum_{m\neq n}a_{mn}(\Pi)~,
\end{equation}
for all $n\in I$. The constraints thus reduce to the single relation \eqref{eq:Cond2} which can additionally be restricted to sets of rank-one projections $\Pi$ by virtue of linearity. Evaluating $a_{mn}(\Pi)$ for $m\neq n$ and rank-one projections associated to an orthonormal basis $\{|n\rangle\}$ of $\mathcal{H}$ one finds
\begin{align}\label{eq:amn}
 a_{mn}(\Pi)&=\mathrm{Tr}\left\{|m\rangle\langle m| \mathcal{K}_t(|n\rangle\langle n|)\right\}\nonumber\\
&=\sum_{j}\gamma_j(t)|\langle m|A_j(t)|n\rangle|^2~,
\end{align}
which is condition \eqref{eq:CondPdiv}.
\end{proof}
Obviously, the conditions for P- and CP-divisibility coincide for a master equation with a single decay channel. We may alternatively characterize P-divisibility employing the contraction property \eqref{eq:contraction} and link it to quantum Markovian behavior defined above (see definition \ref{def:Markov}). To this end, we first prove the following lemma.

\begin{mylem}\label{lem:Helstrom}
For any Hermitian operator $X\neq 0$, there exists a real number $\lambda >0$ and a Helstrom matrix $\Delta$ such that $X=\lambda\Delta$.
\end{mylem}
\begin{proof}
 Let $X=X^\dag\neq0$ be given. If $X\geq0$, then $\rho_1=(\mathrm{Tr}X)^{-1}X$ defines a state so that for $\lambda=\mathrm{Tr}X$ the Helstrom matrix characterized by $p_1=1$, $p_2=0$ with arbitrary $\rho_2$ proves the claim. Similarly for $X\leq0$.

Hence, suppose that $X$ is indefinite. Employing the Jordan-Hahn decomposition we thus find nonzero operators $Y_{1,2}\geq0$ where $Y_1\perp Y_2$ such that $X=Y_1-Y_2$ and, therefore, $\mathrm{Tr}|X|=\mathrm{Tr}Y_1+\mathrm{Tr}Y_2>0$. Clearly, the operators $\rho_i=(\mathrm{Tr}Y_i)^{-1}Y_i$ define states and we have
\begin{equation}
 X=\lambda(p_1\rho_1-p_2\rho_2)~,
\end{equation}
for $\lambda=\mathrm{Tr}|X|$ and $p_i=\lambda^{-1}\mathrm{Tr}Y_i$. These quantities are indeed positive and sum to one, thus, representing a probability distribution which concludes the proof. 
\end{proof}
Since dynamical maps $\Phi_t$ as well as the trace norm are homogeneous (with respect to positive numbers) it thus suffices to apply the characterization of positivity \eqref{eq:contraction} to Helstrom matrices which yields:
\begin{mythe}
 If the dynamical maps defining a process $\Phi$ are bijective, then $\Phi$ is Markovian if and only if it is P-divisible.
\end{mythe}
\begin{proof}
 We first note that $p_1\Phi_t(\rho_1)-p_2\Phi_t(\rho_2)=\Phi_t(\Delta)$ holds for any probability distribution $\{p_i\}$ and pair of states $\rho_i$. Hence, it suffices to consider the time evolution of Helstrom matrices when studying quantum Markovianity. Suppose $\Phi$ is P-divisible. It follows that
\begin{equation}
 \norm{\Phi_t(\Delta)}_1=\norm{\Phi_{t,s}(\Phi_s(\Delta))}_1\leq\norm{\Phi_s(\Delta)}_1~,
\end{equation}
for all $t\geq s\geq0$ and Helstrom matrices $\Delta$ due to positivity of $\Phi_{t,s}$. Hence, $\norm{\Phi_t(\Delta)}_1$ is a monotonically decreasing function of time for any $\Delta$ showing that the process $\Phi$ is Markovian.

For the converse, we first note that the inverse map $\Phi_t^{-1}$ exists for all $t$ as the dynamical maps are bijective on the set of Hermitian operators by assumption. Thus, the maps $\Phi_{t,s}$ exists for all $t\geq s\geq 0$. Now, let the process be Markovian, i.e., $\norm{\Phi_t(\Delta)}_1\leq \norm{\Phi_s(\Delta)}_1$ for all $t\geq s\geq 0$ and Helstrom matrices $\Delta$. We may rewrite this as
\begin{align}
\norm{\Phi_{t,s}(\Phi_s(\Delta))}_1\leq\norm{\Phi_s(\Delta)}_1~.
\end{align}
from which positivity of $\Phi_{t,s}$ follows according to Eq.\,\eqref{eq:contraction} employing lemma \ref{lem:Helstrom} and the fact that $\Phi_t$ is bijective for all $t$. Hence, $\Phi$ is P-divisible. 
\end{proof}

\subsection{Connection to classical Markovian stochastic processes}
\label{sec:conn-class-mark}

The definition of Markovianity stated above is thus equivalent to P-divisibility of a dynamical process $\Phi$ if this notion is well-defined at all. However, the measure for non-Markovianity \eqref{eq:newM} may be evaluated for any dynamical process showing its great benefit. Similarly, one could obtain equivalence of CP-divisibility with Markovian behavior if the dilated process $\Phi\otimes\mathbb{1}_\mathcal{H}$ is considered as shown in Ref. \cite{ChruscinskiRivas}. However, P-divisible quantum processes offer the relevant feature of a distinct connection to classical Markovian stochastic processes. To show this, we employ the characterization of P-divisible processes given in theorem \ref{the:Divisibility}. Let $\rho(t)$ be the solution of the master equation
\begin{equation}\label{eq:MEq}
 \frac{\mathrm{d}}{\mathrm{d}t}\rho(t)=\mathcal{K}_t\rho(t)~,
\end{equation}
with initial state $\rho(0)$ where $\mathcal{K}_t$ is given by Eq.\,\eqref{eq:MQ}. The time-evolved state admits an instantaneous spectral decomposition,
\begin{align}
\rho(t)=\sum_m p_m(t)|\phi_m(t)\rangle\langle\phi_m(t)|~,
\end{align}
so that $\{|\phi_m(t)\rangle\}$ defines an orthonormal basis on $\mathcal{H}$ and $\{p_m(t)\}$ represents a classical probability distribution for all $t\geq0$. By virtue of the orthonormality the eigenvalues $p_m(t)=\langle \phi_m(t)|\rho(t)|\phi_m(t)\rangle$ obey the following closed differential equation:
\begin{align}\label{eq:PauliMQ1}
 \frac{\mathrm{d}}{\mathrm{d}t}p_m(t)=&\sum_{n}
 \Big[W_{mn}(t)p_n(t)-W_{nm}(t)p_m(t)\Big]~,
\end{align}
where 
\begin{equation}\label{eq:transRate}
 W_{mn}(t)=\sum_j \gamma_j(t) |\langle\phi_m(t)|A_j(t)|\phi_n(t)\rangle|^2~.
\end{equation}
Obviously, the term in Eq.\,\eqref{eq:PauliMQ1} with $m=n$ drops out. Given any solution of a quantum master equation, one thus obtains a classical jump process. We emphasize that the rates $W_{mn}(t)$ depend in general on the initial eigenbasis $\{|\phi_m(0)\rangle\}$ as well as the initial probability distribution $\{p_m(0)\}$.

According to the theory of classical stochastic processes, any hierarchy of $n$-point probability distributions $P_n(y_n,t_n;\dots;y_1,t_1)$ with discrete sample space $\Omega$ satisfying the consistency relations
 \begin{align}
 &P_n(y_n,t_n;\dots;y_1,t_1)\geq0~,\label{eq:hier1}\\
 \sum_{y_1\in\Omega} &P_1(y_1,t_1)=1~,\label{eq:hier3}\\
 \sum_{y_m\in\Omega} &P_n(y_n,t_n;\dots;y_m,t_m;\dots;y_1,t_1)\nonumber\\
&=P_{n-1}(y_n,t_n;\dots;y_1,t_1)~,\label{eq:hier2}
 \end{align}
for all $1\leq m\leq n$ and $0\leq t_1<t_2<\dots<t_n$, determines a stochastic process $\mathcal{Y}(t)$ with values in $\Omega$ \cite{Kolmogorov,Gardiner,vKampen}. Here, $P_n(y_n,t_n;\dots;y_1,t_1)$ gives the probability to observe the values $y_i$ at times $t_i$ for $i=1,\dots,n$ for the stochastic process $\mathcal{Y}(t)$. For classical Markov processes obeying the relation
\begin{align}\label{eq:classMarkov}
 &P_{1|n-1}(y_n,t_n|y_{n-1},t_{n-1};\dots;y_1,t_1)\nonumber\\
 &=P_{1|1}(y_n,t_n|y_{n-1},t_{n-1})~,
 \end{align}
for any conditional probability distribution $P_{1|n-1}$, the full hierarchy is completely determined by the $1$-point probability distribution $P_1$ and the so-called transition probability $P_{1|1}$. However, these two nonnegative functions cannot be chosen arbitrarily but must satisfy
 \begin{align}
  P_1(y_2,t_2)&=\sum_{y_1\in\Omega}P_{1|1}(y_2,t_2|y_1,t_1)P_{1}(y_1,t_1)~,\label{eq:KonsistMarkov}\\
P_{1|1}(y_3,t_3&|y_1,t_1)\label{eq:Chapman-Markov}\\
 &=\sum_{y_2\in\Omega}P_{1|1}(y_3,t_3|y_2,t_2)P_{1|1}(y_2,t_2|y_1,t_1)~,\nonumber
 \end{align}
 where Eq.\,\eqref{eq:Chapman-Markov} is called the Chapman-Kolmogorov equation. These equations thus characterize uniquely  a Markov process. 

One may equivalently write this relation for the transition probability in its differential form
\begin{align}
 &\frac{\mathrm{d}}{\mathrm{d}t}P_{1|1}(y,t|x,s)\\
&=\sum_{z\in\Omega}
\Big[ \mathcal{W}_{yz}(t)P_{1|1}(z,t|x,s)-\mathcal{W}_{zy}(t)P_{1|1}(y,t|x,s) \Big]~,\nonumber
\end{align}
where the transition probability per unit time $\mathcal{W}_{yz}(t)$ is nonnegative and represents the probability for a transition to $y$ given the classical state was $z$ at time $t$. It is clear that the (differential) Chapman-Kolmogorov equation characterizes the transition probability $P_{1|1}$ but a similar equation, the so-called Pauli master equation, can be derived for the $1$-point probability distribution, that is,
\begin{align}
 &\frac{\mathrm{d}}{\mathrm{d}t}P_{1}(y,t)\\
&=\sum_{z\in\Omega}
 \Big[ \mathcal{W}_{yz}(t)P_{1}(z,t)-\mathcal{W}_{zy}(t)P_{1}(y,t) \Big]~.\nonumber
\end{align}
One thus concludes that Eq.\,\eqref{eq:PauliMQ1} can be interpreted as Pauli master equation for the $1$-point probability distribution of a classical Markov process with $\Omega=\{1,\dots,\mathrm{dim}\mathcal{H}\}$ if and only if
\begin{equation}\label{eq:transRate2}
 W_{mn}(t)=\sum_j \gamma_j(t) |\langle\phi_m(t)|A_j(t)|\phi_n(t)\rangle|^2\geq0~,
\end{equation}
for all $t\geq0$ and $m\neq n$. Hence, P-divisibility of the quantum process is a sufficient condition to warrant positivity of the rates $W_{mn}(t)$ (see Eq.\,\eqref{eq:CondPdiv}). This shows that quantum non-Markovianity defined with respect to the generalized trace distance based measure does not only provide an interpretation in terms of an information backflow but also allows for a connection to classical Markovian stochastic processes. Namely, to each P-divisible quantum process, given as the solution of a master equation of the form \eqref{eq:MQ} with Lindblad operators and rates satisfying Eq.\,\eqref{eq:CondPdiv} for an initial state $\rho(0)$, one associates a classical Markovian process obtained as solution of the classical master equation \eqref{eq:PauliMQ1} with transition rates given by Eq.\,\eqref{eq:transRate} and initial condition specified by the eigenvalues of $\rho(0)$.

Finally, we note that P-divisibility of the quantum process is indeed equivalent to the positivity of the rates $W_{mn}(t)$ of the classical Pauli master equation if the quantum master equation has the property that the eigenbases $\{|\phi_n(t)\rangle\}$ of $\rho(t)$ run over all orthonormal bases when varying the initial state $\rho(0)$. This is satisfied if the maximally mixed state is in the image of the dynamical map $\Phi_t$ which holds true for sufficiently small times due to continuity of the process and as we have $\Phi_{0}=\mathbb{1}$. For two-level systems it is shown in Appendix \ref{app:maxMix} that this constraint is not only sufficient but also necessary. If any orthonormal basis is encountered all classical processes derived from the quantum master equation are classically Markovian if and only if the quantum dynamics is P-divisible.

\section{Generalization of the trace distance based non-Markovianity measure}
\label{sec:gener-trace-dist}

Having stated the generalized definition for quantum non-Markovianity which allows for a connection to classical Markov processes, we now consider the expression of the corresponding measure and address its mathematical and physical features.
According to the definition of a quantum Markovian process given in Sec.~\ref{sec:ensemble-discr}, 
it is natural to consider the following measure for quantum non-Markovianity, quantifying the degree of memory effects with respect to this generalized definition \cite{ReviewHP}
\begin{align}\label{eq:newM}
 \mathcal{N}(\Phi)&\equiv\max_{\{p_i,\rho_i\}}\int_{\sigma>0} \mathrm{d}t~ \sigma(t,p_i,\rho_i)\\
\intertext{with}
\sigma(t,p_i,\rho_i)&\equiv\frac{\mathrm{d}}{\mathrm{d}t}\norm{p_1\Phi_{t}(\rho_1)-p_2\Phi_{t}(\rho_2)}_1~,
\end{align}
where the integration runs over all intervals where the distinguishability $\norm{p_1\Phi_{t}(\rho_1)-p_2\Phi_{t}(\rho_2)}_1$ increases. A process $\Phi$ is said to be non-Markovian if $\mathcal{N}(\Phi)>0$. As discussed in Sec.~\ref{sec:ensemble-discr} this measure still admits the interpretation as a quantifier for the information flow from the environment back to the open 
system. In addition, as discussed in Sec.~\ref{sec:dist-divis}, Markovianity will now be equivalent to P-divisibility of a quantum process. Note that this feature was only sufficient in the original definition (see Section \ref{sec:examples}).

To begin with, we concentrate on the maximization procedure contained in the quantifier for non-Markovianity \eqref{eq:newM} which, contrary to the original definition \cite{BLP}, now even requires to sample also over binary probability distributions. Fortunately, a similar characterization of pairs of states maximizing Eq.\,\eqref{eq:newM} can be proven \cite{OptimalStates} and yet also a local representation is admitted \cite{LocalRep} simplifying the sampling significantly.

We finally illustrate the novel definition of non-Markovianity by means of examples which show the difference to the original characterization and other approaches to non-Markovianity.

\subsection{Expressions of the generalized non-Markovianity measure}
\label{sec:expr-gener-non}
A set $\{p_i,\rho_{1,2}\}$, where $\rho_{1,2}\in\mathcal{S}(\mathcal{H})$ and $\{p_i\}$ defines a binary probability distribution, is said to be \emph{optimal} if the maximum in Eq.\,\eqref{eq:newM} is attained for it.
Note that the quantum states of an optimal set are necessarily nonequal as the dynamics of the norm of a nonindefinite Helstrom matrix under any dynamical map is trivial due to trace preservation and positivity of the map.

\begin{mythe}\label{the:orthogonal}
The states of an optimal set must be orthogonal, i.e.
\begin{align}\label{eq:newM2}
 \mathcal{N}(\Phi)&=\max_{\{p_i,\rho_1\perp\rho_2\}}\int_{\sigma>0} \mathrm{d}t~ \sigma(t,p_i,\rho_i)~.
\end{align}
\end{mythe}
\begin{proof}
Suppose $\{p_i,\rho_{1,2}\}$ be optimal with $\rho_1\not\perp\rho_2$. Along the lines of the proof of lemma \ref{lem:Helstrom} we obtain a probability distribution $\{q_i\}$ and two orthogonal states $\varrho_{1,2}$ such that
\begin{equation}
 p_1\rho_1-p_2\rho_2=\lambda(q_1\varrho_1-q_2\varrho_2)~,
\end{equation}
where $\lambda=\mathrm{Tr}|p_1\rho_1-p_2\rho_2|$. As $\rho_1\not\perp\rho_2$, we have $0<\lambda<1$ according to Eq.~\eqref{eq:boundHelstrom}. Here, $\lambda>0$ holds as states of an optimal set are nonequal by definition. By means of linearity of the dynamical maps $\Phi_t$ and homogeneity of the trace norm one finally obtains
\begin{equation}
  \norm{q_1\varrho_1(t)-q_2\varrho_2(t)}_1=\frac{1}{\lambda}\norm{p_1\rho_1(t)-p_2\rho_2(t)}_1
\end{equation}
for all $t\geq0$ where $\lambda^{-1}>1$. This shows that any increase of $\norm{p_1\rho_1(t)-p_2\rho_2(t)}_1$ is exceeded by the increase of $\norm{q_1\varrho_1(t)-q_2\varrho_2(t)}_1$. Hence, $\{q_i,\varrho_{1,2}\}$ yields a non-Markovianity strictly larger than the set $\{p_i,\rho_{1,2}\}$ contradicting its optimality.
\end{proof}

We define $\mathring{\mathcal{S}}(\mathcal{H})$ to be the interior of the state space, i.e. the set of all quantum states $\rho$ for which there is an $\epsilon>0$ such that all Hermitian operators $X$ with unit trace satisfying $||\rho-X||_1\leq\epsilon$ belong to $\mathcal{S}(\mathcal{H})$. Hence, all states in the interior have full rank. Note that $\mathring{\mathcal{S}}(\mathcal{H})=\emptyset$ if $\mathrm{dim}\mathcal{H}=\infty$ which implies that the local representation of the measure \eqref{eq:newM} introduced below is not available in infinite dimensions. However, these quantum systems can frequently be accurately described by finite-dimensional Hilbert spaces.

Based on the orthogonality of states of an optimal set we can also establish a local representation for the generalized measure as provided for the original definition in Ref. \cite{LocalRep}. To this end, one first has to prove a characterization of enclosing surfaces which are defined as follows:
 A set $\partial U(\rho)\subset\mathcal{S}(\mathcal{H})$ not containing $\rho\in\mathring{\mathcal{S}}(\mathcal{H})$ is called an \emph{enclosing surface} of $\rho$ if and only if for any nonzero, Hermitian and traceless operator $Y$ there exists a real number $\mu>0$ such that
\begin{equation}\label{eq:defEncl}
 \rho+2\mu Y\in\partial U(\rho)~.
\end{equation}

\begin{figure}[hb]
    \centering
     \includegraphics[width=4cm]{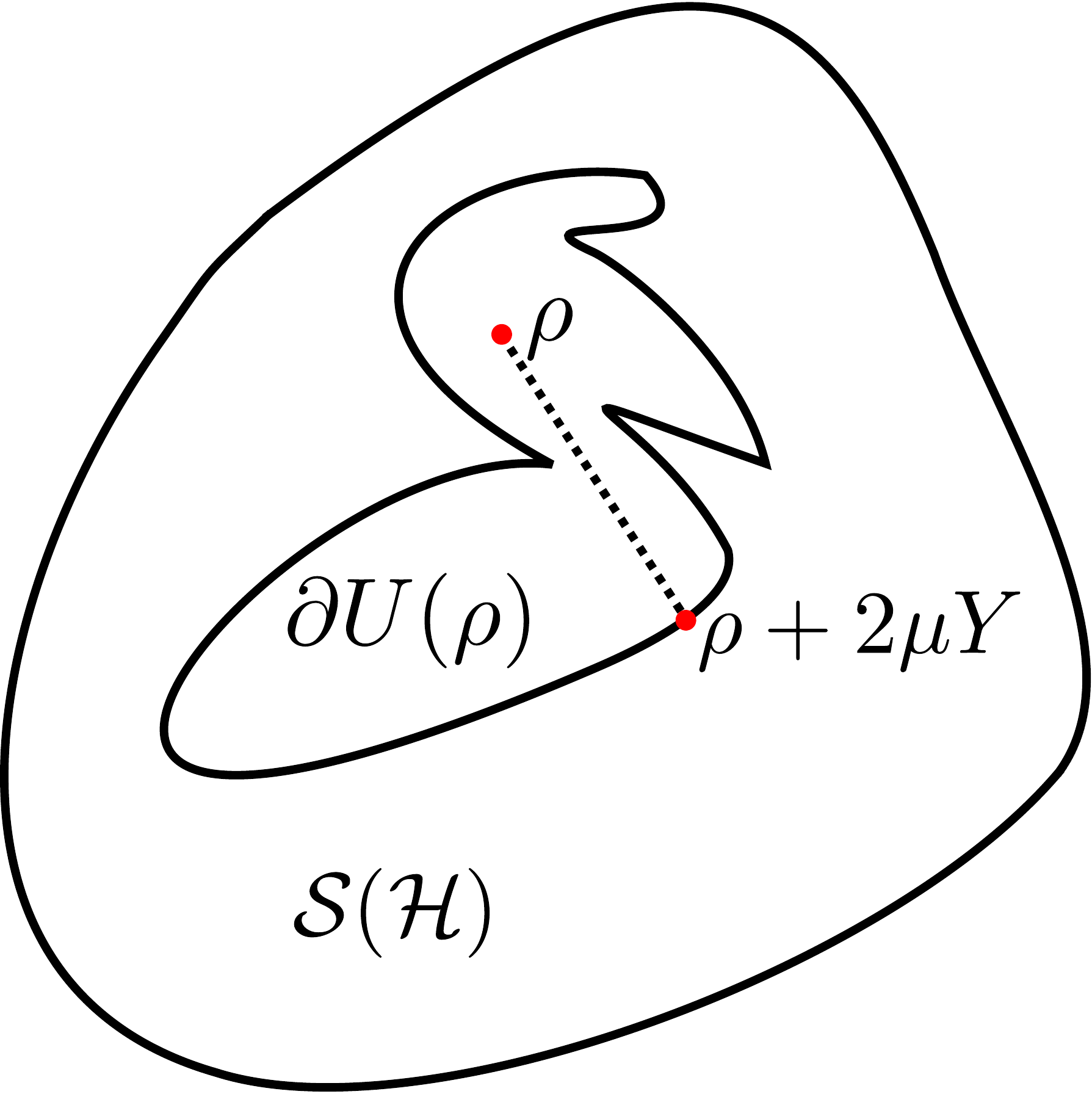} 
     \caption{Illustration of an enclosing surface $\partial U(\rho)$ of an inner point $\rho$ of the state space $\mathcal{S}(\mathcal{H})$.}\label{fig:EnclosingSurf}
\end{figure}

\begin{mylem}\label{lem:equivS}
 Let $\partial U(\rho)$ be an enclosing surface and $X$ a nonzero, Hermitian and indefinite operator. Then there exists a real number $\lambda>0$ with $\lambda|\mathrm{Tr}X|<1$ such that 
\begin{align}\label{eq:defNEncl}
 \frac{1}{p_-}(p_+\rho-\mathrm{sgn}(\mathrm{Tr}X)\lambda X)\in\partial U(\rho)~,
\end{align}
where~ $p_\pm\equiv\frac{1}{2}(1\pm\lambda|\mathrm{Tr}X|)$ and
\begin{equation}
\mathrm{sgn}(x)=\begin{cases}
                -1~,&~\text{if}~x\leq0\\ +1~,&~\text{else}
               \end{cases}
\end{equation}
\end{mylem}

\begin{proof}
Let $\partial U(\rho)$ be an enclosing surface of $\rho$. Consider a nonzero, Hermitian, indefinite operator $X$. The operator 
\begin{equation}
 Y= \mathrm{sgn}(\mathrm{Tr}X)\bigl[(\mathrm{Tr}X)\rho-X\bigr]
\end{equation}
defines a nonzero, Hermitian and traceless operator. By definition, there exists a real number $\mu>0$ such that
\begin{equation}
 \varrho=\rho+2\mu Y\in\partial U(\rho)~.
\end{equation}
We then define a real number $\lambda$ by means of 
\begin{equation}
 \mu=\frac{\lambda}{1-\lambda|\mathrm{Tr}X|}
\end{equation}
so that $\lambda=\mu/(1+\mu|\mathrm{Tr}X|)>0$ holds and hence $\lambda|\mathrm{Tr}X|<1$. For $p_\pm=\frac{1}{2}(1\pm\lambda|\mathrm{Tr}X|)$
one finally obtains
\begin{align}
 \varrho&=\rho+\frac{2\lambda}{1-\lambda|\mathrm{Tr}X|}\left\{\mathrm{sgn}(\mathrm{Tr}X)\bigl[(\mathrm{Tr}X)\rho-X\bigr]\right\}\nonumber\\
&=\rho-\frac{\lambda}{p_-}\mathrm{sgn}(\mathrm{Tr}X) X+\frac{p_+-p_-}{p_-}\rho \nonumber\\
&=\frac{1}{p_-}(p_+\rho-\mathrm{sgn}(\mathrm{Tr}X)\lambda X)~,
\end{align}
which is Eq.\,\eqref{eq:defNEncl}.
\end{proof}
As proven in Appendix \ref{app:enclSurf}, lemma \ref{lem:equivS} can actually be augmented to show that the characterization \eqref{eq:defNEncl} of an enclosing surface is equivalent to Eq.\,\eqref{eq:defEncl}. 
\begin{mythe}\label{the:locRep1}
 The generalized measure of quantum non-Markovianity admits a local representation, i.e. 
\begin{align}\label{eq:newMloc}
 \mathcal{N}(\Phi)&=\max_{\{p_i\},\rho_2\in\partial U(\rho_1)}\int_{\sigma>0} \mathrm{d}t~ \tilde{\sigma}(t,p_i,\rho_i)\\
\intertext{with}
\tilde{\sigma}(t,p_i,\rho_i)\equiv&\frac{1}{\norm{p_1\rho_1-p_2\rho_2}_1}\frac{\mathrm{d}}{\mathrm{d}t}\norm{p_1\Phi_t(\rho_1)-p_2\Phi_t(\rho_2)}_1~,\label{eq:sigmatilde}
\end{align}
where $\rho_1\in\mathring{\mathcal{S}}(\mathcal{H})$ is any fixed inner point of the state space and  $\partial U(\rho_1)$ refers to an arbitrary enclosing surface of $\rho_1$.
\end{mythe}
We emphasize that $\tilde{\sigma}(t,p_i,\rho_i)$ is well-defined as $0<\norm{p_1\rho_1-p_2\rho_2}_1$ for any state $\rho_2\in\partial U(\rho_1)$ and probability distribution $\{p_i\}$ as by definition of an enclosing surface we have $\rho_2\neq\rho_1$ (cf. Eq.\,\eqref{eq:boundHelstrom}).
\begin{proof}
Let $\partial U(\rho_1)$ be an enclosing surface of $\rho_1$. To prove that the corresponding local representation yields a value smaller than or equal to the original definition \eqref{eq:newM2} one follows the lines of the proof of theorem \ref{the:orthogonal}. Let $\rho_2\in\partial U(\rho_1)$ and a probability distribution $\{p_i\}$ be given. Now, according to lemma \ref{lem:Helstrom} there exist two state $\varrho_1\perp\varrho_2$ and a probability distribution $\{q_i\}$ such that
\begin{align}
 p_1\rho_1-p_2\rho_2=\lambda(q_1\varrho_1-q_2\varrho_2)~,
\end{align}
where $\lambda=\norm{p_1\rho_1-p_2\rho_2}_1>0$. We then obtain 
\begin{align}
 \tilde{\sigma}(t,p_i,\rho_i)=\sigma(t,q_i,\varrho_i)~,
\end{align}
for all $t\geq0$ due to linearity of $\Phi_t$ and homogeneity of the trace norm and the derivative. We conclude that the right-hand side of Eq.\,\eqref{eq:newMloc} is smaller than or equal to $\mathcal{N}(\Phi)$ as defined in Eq.\,\eqref{eq:newM2}.

Conversely, let $\varrho_1\perp\varrho_2$ be two orthogonal states and denote by $\{q_i\}$ a probability distribution. Then $\Delta= q_1\varrho_1-q_2\varrho_2$ defines a nonzero, Hermitian, indefinite operator. Thus, according to lemma \ref{lem:equivS}, there exists a real number $\lambda>0$ with $\lambda|\mathrm{Tr}\Delta|<1$ such that
\begin{equation}
 p_+\rho_1-p_-\rho_2=c\lambda \Delta
\end{equation}
where $c=\mathrm{sgn}(\mathrm{Tr}\Delta)$ and $p_\pm=\tfrac{1}{2}(1\pm\lambda|\mathrm{Tr}\Delta|)$ for some quantum state $\rho_2\in\partial U(\rho_1)$ of the enclosing surface. As $|c|=1$ it follows that $\norm{p_+\rho_1-p_-\rho_2}_1=\lambda>0$. Linearity of the dynamical map and homogeneity of the trace norm and the derivative then yields
\begin{align}
\tilde{\sigma}(t,p_\pm,\rho_i)=\sigma(t,q_i,\varrho_i)~,
\end{align}
showing that the original definition \eqref{eq:newM2} leads to a value which is smaller than or equal to the right-hand side of Eq.\,\eqref{eq:newMloc}. This thus concludes the proof, that is, the maximization over an enclosing surface with an information flux rescaled by the initial distinguishability reproduces the dynamics of the trace norm of orthogonal states.  
\end{proof}

A careful reformulation of an enclosing surface thus allows to establish an equivalent local and universal representation of the generalized trace distance based measure as for the original definition \cite{LocalRep}. This shows that also for the novel characterization, non-Markovianity is a universal feature appearing everywhere in state space. Similarly, it suffices if Eq.\,\eqref{eq:defEncl} holds for exactly one $\mu>0$ as $\lambda$ is uniquely determined by this parameter given a nonzero, Hermitian and indefinite operator (cf. proof of lemma \ref{lem:equivS}). In addition, also here, no assumption on the enclosing surface concerning for example the shape or the smoothness is needed implying a great benefit for the analytical, numerical and experimental determination of the generalized measure.

\subsection{Examples}\label{sec:examples}
As already stated before the fundamental difference between the trace distance based measure for non-Markovianity and its generalization is that in the trace distance approach P-divisibility is only a sufficient criterion for Markovianity. This fact becomes particularly apparent when uniform translations of states are encountered in the dynamics which do not describe positive maps but leave the trace distance unchanged \cite{NonUnitalNM}. However, choosing unequal weights $p_i$ the Helstrom matrices are no longer invariant under such operations thus making it possible to detect the effect of such maps. 

To illustrate this property, we consider a two-level open quantum system, i.e. $\mathcal{H}=\mathbb{C}^2$, undergoing a dynamics which, as depicted in Fig.\,\ref{fig:BlochSphere}, can be described in the Bloch sphere as an isotropic contraction followed by a translation along the $z$-axis. The dynamics is analytically described by a time-local master equation \eqref{eq:MEq} where the generator $\mathcal{K}_t$ obeys
\begin{align}\label{eq:ExpGen1}
 \mathcal{K}_t\rho(t)=\sum_{j=1}^3\frac{\gamma(t)}{4}\left[\sigma_j\rho(t)\sigma_j-\rho(t)\right]~,
\end{align}
for $0\leq t\leq t_1$ with $t_1<T$ and
\begin{align}\label{eq:ExpGen2}
 \mathcal{K}_t\rho(t)=&\frac{-b(t)}{2}\left[\sigma_-\rho(t)\sigma_+-\tfrac{1}{2}\{\sigma_+\sigma_-,\rho(t)\}\right]\nonumber\\
&+\frac{b(t)}{2}\left[\sigma_+\rho(t)\sigma_--\tfrac{1}{2}\{\sigma_-\sigma_+,\rho(t)\}\right]~,
\end{align}
if $t_1\leq t\leq T$ holds. Here, $\sigma_j$ refers to the Pauli matrices and $\sigma_{+(-)}$ describes the usual raising (lowering) operator with respect to the eigenstates $|\pm\rangle$ of $\sigma_3$. Moreover, the rates satisfy $\gamma(t), b(t)\geq0$ for all $t$ in their respective domains so that the second phase of the process is neither CP- nor P-divisible which can be easily shown using theorem \ref{the:Divisibility}. This dynamical process exhibits the essential feature of the generalized amplitude damping channel studied in Ref. \cite{NonUnitalNM}.

\begin{figure}[hb]
    \centering
     \subfloat[]{{\includegraphics[width=2.45cm]{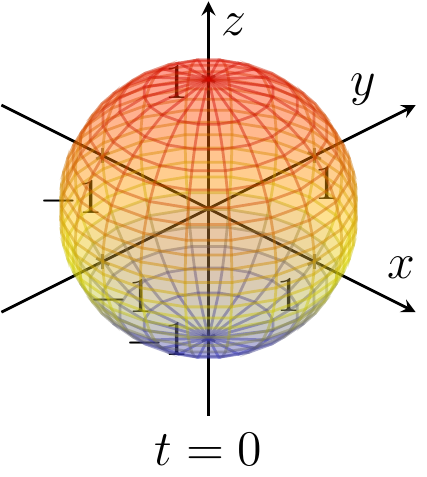} }}$\rightarrow$
     \subfloat[]{{\includegraphics[width=2.45cm]{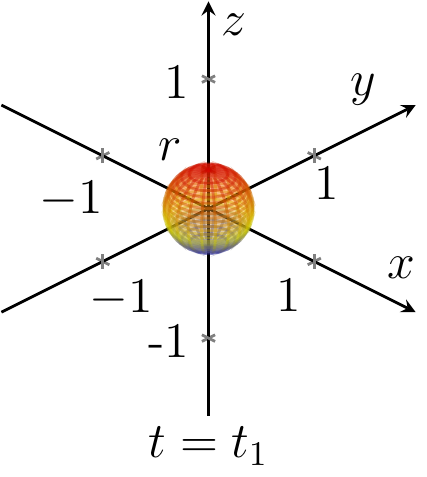} }}$\rightarrow$
     \subfloat[]{{\includegraphics[width=2.45cm]{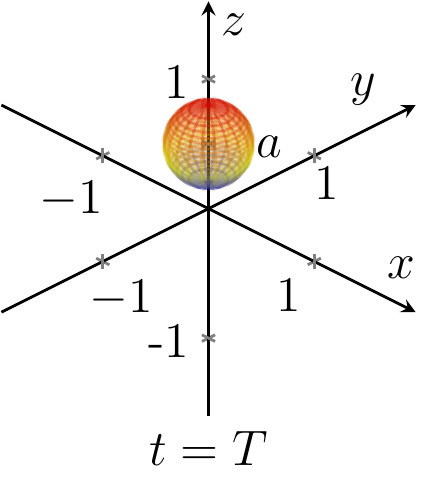} }}
     \caption{Bloch sphere representation of the action of the dynamical map $\Phi_t$ corresponding to the time-local generator $\mathcal{K}_t$ defined by Eqs.\,\eqref{eq:ExpGen1} and \eqref{eq:ExpGen2} for $t=0,t_1$ and $T$.}\label{fig:BlochSphere}
\end{figure}

Employing the Bloch vector representation for two-level systems the master equation is equivalently described by a differential equation for the Bloch vector,
\begin{equation}
 \frac{\mathrm{d}}{\mathrm{d}t}\vec{v}(t)=\begin{cases}
                                           A(t)\vec{v}(t)~,&~0\leq t\leq t_1<T\\ \vec{b}(t)~,&~t_1\leq t\leq T
                                          \end{cases}
\end{equation}
where $A(t)=\mathrm{diag}(-\gamma(t),-\gamma(t),-\gamma(t))$ and $\vec{b}(t)=(0,0,b(t))^T$. Hence, the process' first phase corresponds to an isotropic contraction of the Bloch sphere $B_1=\{\vec{v}\,|\,|\vec{v}|\leq1\}$ to the smaller sphere $B_r=\{\vec{v}\,|\,|\vec{v}|\leq r\}$ with radius
\begin{equation}
 r=\exp\left[-\int_0^{t_1}\mathrm{d}t\, \gamma(t)\right]\in(0,1)~,
\end{equation}
which is clearly CP-divisible and, therefore, Markovian. On the other hand, the second phase describes a uniform translation of the Bloch sphere along the $z$-axis, i.e. $\vec{v}(t_1)\mapsto\vec{v}(t_1)+(0,0,a)^T$ with
\begin{equation}
 a=\int_{t_1}^{T}\mathrm{d}t\, b(t)>0~,
\end{equation}
where we thus have to require $a\leq 1-r$ in order to maintain positivity of the dynamical map (see Fig.\,\ref{fig:BlochSphere}). It is easily shown that this condition is also necessary and sufficient for complete positivity of the process. 

The trace norm of the Helstrom matrix $\Delta=p_1\rho_1-p_2\rho_2$ at time $t$ for two quantum states $\rho_{1,2}=\tfrac{1}{2}(\mathbb{1}_2+\vec{v}_{1,2}\cdot\vec{\sigma})$, evolving according to the described dynamical map, is given by
\begin{equation}\label{eq:DeltaExp}
 \norm{\Delta(t)}_1=\frac{1}{2}\left\{|p_1-p_2+|\vec{w}(t)|\,|+|p_1-p_2-|\vec{w}(t)|\,|\right\}~,
\end{equation}
where $\vec{w}(t)=p_1\vec{v}_1(t)-p_2\vec{v}_2(t)$. Clearly, if $p_1=p_2=1/2$, then $\norm{\Delta(t)}_1=\tfrac{1}{2}|\vec{v}_1(t)-\vec{v}_2(t)|$ showing that the unbiased case is unable to detect the non-Markovianity of the process resulting from the uniform translation.

Without restriction we may assume $p_1\geq p_2$ so that
\begin{equation}\label{eq:DeltaExp2}
 \norm{\Delta(t)}_1=\begin{cases}
                     p_1-p_2~,~&\text{if}~ p_1-p_2>|\vec{w}(t)|\\
		    |\vec{w}(t)|~,~&\text{if}~ p_1-p_2\leq|\vec{w}(t)|~.\\
                    \end{cases}
\end{equation}
and, due to theorem \ref{the:orthogonal}, we may restrict ourselves to orthogonal states corresponding to antipodal unit vectors. That is, $\vec{v}_1=-\vec{v}_2$ with $|\vec{v}_1|=1$, which implies $\vec{w}(t)=\vec{v}_1(t)$.

If the probability distribution is such that $p_1-p_2>|\vec{v}_1(t)|$ holds for all $t_1\leq t\leq T$, then $\norm{\Delta(t)}_1$ is a monotonically decreasing function of $t$. However, if $p_1-p_2<|\vec{v}_1(T)|$, then 
\begin{equation}
 \int_{\sigma>0}\mathrm{d}t\,\sigma(t)=|\vec{v}_1(T)|-(p_1-p_2)>0~,
\end{equation}
indicating non-Markovianity. The change of the trace norm thus increases with decreasing difference $p_1-p_2$ which is bounded by $p_1-p_2=r=|\vec{v}_1(t_1)|$. Calculating explicitly the change of the trace norm for $p_1-p_2\leq r$ one finds
\begin{align}
 &\norm{\Delta(T)}_1-\norm{\Delta(t_1)}_1\nonumber\\
=&\,|\vec{v}_1(t_1)+(p_1-p_2)\vec{a}|-|\vec{v}_1(t_1)|~,
\end{align}
which shows that the trace norm attains its maximal value 
\begin{align}
 \norm{\Delta(T)}_1-\norm{\Delta(t_1)}_1=r|\vec{a}|~,
\end{align}
if $\vec{v}_1$ is parallel to $\vec{a}$, i.e. $\vec{v}_1=c\,\vec{a}$ for some $c>0$, and $p_1-p_2=r$. Obviously, this is satisfied by the probability distribution $\{p_i\}=\{p_\pm=\tfrac{1}{2}(1\pm r)\}$. Unlike the original definition, the generalized trace distance based measure is thus able to capture non-Markovian dynamics arising from uniform translations contained in the dynamical process.

The dynamics generated by $\mathcal{K}_t$ for $0\leq t\leq t_1$ (cf. Eq.\,\eqref{eq:ExpGen1}) for any $t_1>0$ provides also an example for a process which is not CP- but P-divisible. Choosing $\gamma_1(t)=\gamma_2(t)=1$ and $\gamma_3(t)=-\tanh(t)$ as proposed in Ref. \cite{Hall}, the dynamical process is not CP-divisible according to theorem \ref{the:Divisibility}. In particular, there exists even no single interval for which CP-divisibility is restored as $\gamma_3(t)<0$ for all $t>0$. However, the dynamics is always P-divisible since
\begin{equation}
 \sum_{i=1}^2|\langle m|\sigma_i|n\rangle|^2-\tanh(t)|\langle m|\sigma_3|n\rangle|^2\geq0~,
\end{equation}
is valid for all $0\leq t\leq t_1$ and $m\neq n$ which is condition \eqref{eq:CondPdiv}. Hence, the trace distance based measure and its generalization are equal to zero in this case while measures for non-Markovianity relying on CP-divisibility \cite{RHP,Hall,ChruscinskiRivas} are, of course, nonvanishing. This random unitary evolution \cite{BasanoRandU,ChruRandomU} illustrates the persisting and significant difference between the two major approaches for the characterization of quantum non-Markovianity. 

\section{Conclusions and outlook}
\label{sec:conclusions-outlook}
We have introduced a generalization of the criterion of quantum non-Markovianity based on the flow of information. This novel characterization relies on the trace norm of Helstrom matrices which can also be interpreted as a measure for the distinguishability. By virtue of this property, the generalized measure still admits an interpretation as quantifier of an information backflow from the environment to the open quantum system.

It is shown that the generalized criterion is equivalent to P-divisibility of the dynamical process which has an explicit connection to classical Markovian stochastic processes. That is, any rate equation obtained from a quantum master equation of a P-divisible process can be interpreted as Pauli master equation of a classical Markov process. However, the presented approach is more general since it can also be applied even when the notion of divisibility is ill-defined and can be experimentally tested. 

The experimental determination is substantially simplified by the derived mathematical representations of the generalized measure which are similar to those for the original definition. First, we have demonstrated that optimal initial states for non-Markovian dynamics must be orthogonal and, based on this result, we could finally establish a local representation for the measure. Hence, orthogonal states, corresponding to a maximal information content, exhibit maximal memory effects which can be revealed locally and anywhere in the quantum state space as provided by the local representation.

An essential feature of the generalized approach to non-Markovianity in comparison with the original definition is its sensitivity to memory effects arising from uniform translations of states. To illustrate this, we constructed a dynamical process for a two-level system which comprises a uniform translation of the Bloch sphere. However, there exist dynamical processes which are not CP-divisible but P-divisible as we show in our second example which manifests the existing difference of the generalized definition with other approaches to non-Markovianity.

We believe that the definition of non-Markovianity given in this paper is of great relevance for the study of memory effects in the field of complex quantum systems and quantum information due to the experimental accessibility and its clear-cut interpretation and connection to classical Markov processes.

\acknowledgments HPB and BV acknowledge support from the project EU STREP PROACTIVE H2020 QuProCS (Grant Agreement 641277). BV also acknowledges support by the COST Action MP1006 Fundamental Problems in Quantum Physics and by UniMI through the H2020 Transition Grant 14-6-3008000-623.  SW thanks the German National Academic Foundation for support.

\appendix
\section{Orthonormal bases and the maximally mixed state}\label{app:maxMix}
We prove the statement made in the main text (see Sec. \ref{sec:conn-class-mark}) about the relation between the maximally mixed state and the eigenbases of the set of time-evolved states $\mathrm{Im}\Phi_t=\{\Phi_t(\rho)|\rho\in\mathcal{S}(\mathcal{H})\}$ for two-level systems, i.e. $\mathcal{H}=\mathbb{C}^2$. Note that $\mathrm{Im}\Phi_t$ defines a nonempty, convex and compact set for any finite-dimensional Hilbert space due to linearity of $\Phi_t$ along with compactness and convexity of the state space.

\begin{mylem}
Any orthonormal basis $\{|\Psi_i\rangle\}$ of $\mathcal{H}$ defines the eigenbasis of a quantum state $\rho(t)\in\mathrm{Im}\Phi_t$ if and only if $\tfrac{1}{2}\mathbb{1}_2\in\mathrm{Im}\Phi_t$.
\end{mylem}

\begin{proof}
The 'if' statement is clear from the fact that any orthonormal basis defines a resolution of identity. Hence, if $\tfrac{1}{2}\mathbb{1}_2\in\mathrm{Im}\Phi_t$ then any basis $\{|\Psi_i\rangle\}$ defines at least the eigenbasis of the maximally mixed state proving the claim. It is clear that the same reasoning actually applies to any finite-dimensional Hilbert space.

To show the reverse we employ the Hahn-Banach separation theorem for normed vector spaces. Suppose that $\tfrac{1}{2}\mathbb{1}_2\notin\mathrm{Im}\Phi_t$. Then there exists a real-valued, linear and continuous functional $\varphi_A$ separating the disjoint nonempty, convex and compact sets $\mathrm{Im}\Phi_t$ and $\{\tfrac{1}{2}\mathbb{1}_2\}$. That is, we have
\begin{align}
 \varphi_A(\tfrac{1}{2}\mathbb{1}_2)<\mathrm{inf}\{\varphi_A\bigl(\rho(t)\bigr)|\,\rho(t)\in\mathrm{Im}\Phi_t\}~,
\end{align}
where $\varphi_A(X)=\mathrm{Tr}\{AX\}$ for some Hermitian operator $A$ due to Riesz representation theorem. By means of the transformation $A\mapsto A-\mathrm{Tr}A$ one may also assume that the operator $A$ is traceless. The set $\{X|X=X^\dag,\mathrm{Tr}\{X\}=1,\varphi_A(X)=0\}$ thus describes a hyperplane which intersects the maximally mixed state and separates it from $\mathrm{Im}\Phi_t$. 

Employing the Bloch representation this hyperplane defines an ordinary plane in $\mathbb{R}^3$ which contains the origin and therefore intersects the surface of the Bloch sphere in a unit circle. As pure, orthogonal states on $\mathbb{C}^2$ correspond to antipodal points on the surface of the Bloch sphere we thus conclude that the (hyper)plane contains an uncountable set of orthonormal bases which do not define eigenbases of quantum states $\rho(t)\in\mathrm{Im}\Phi_t$ (cf. Fig. \ref{fig:Hyperplane}). 

\end{proof}

\begin{figure}[t]
    \centering
     \includegraphics[width=4.5cm]{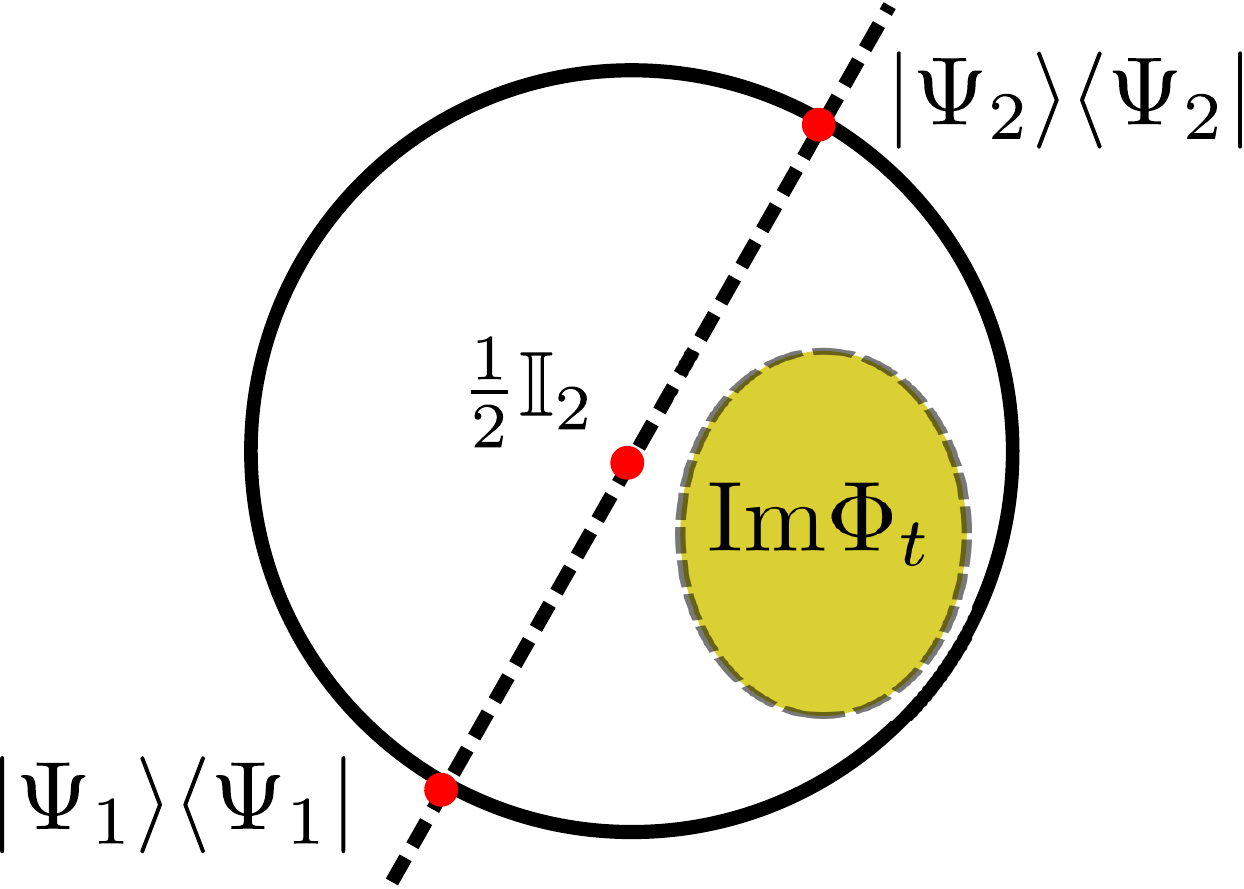} 
     \caption{Two-dimensional cut of a hyperplane in the Bloch representation separating the two disjoint nonempty, convex and compact sets $\{\tfrac{1}{2}\mathbb{1}_2\}$ and $\mathrm{Im}\Phi_t$. As the plane intersects the maximally mixed state its intersection with the state space comprise an uncountable set of pure and orthogonal states defining orthonormal bases of $\mathbb{C}^2$.}\label{fig:Hyperplane}
\end{figure}

\section{Characterizations of enclosing surfaces }\label{app:enclSurf}
In this section we show that the relation for enclosing surfaces in terms of nonzero, Hermitian and indefinite operators derived in lemma \ref{lem:equivS} is equivalent to the original definition, i.e.

\begin{mylem}\label{lem:equivSA}
 A set $\partial U(\rho)$ defines an enclosing surface  of $\rho$ if and only if for any nonzero, Hermitian, indefinite operator $X$ there exist a real number $\lambda>0$ with $\lambda|\mathrm{Tr}X|<1$ such that 
\begin{align}\label{eq:defNEnclA}
 \frac{1}{p_-}(p_+\rho-\mathrm{sgn}(\mathrm{Tr}X)\lambda X)\in\partial U(\rho)~,
\end{align}
where~ $p_\pm\equiv\frac{1}{2}(1\pm\lambda|\mathrm{Tr}X|)$ and
\begin{equation}
\mathrm{sgn}(x)=\begin{cases}
                -1~,&~\text{if}~x\leq0\\ +1~,&~\text{else}
               \end{cases}
\end{equation}
\end{mylem}

\begin{proof}
That any enclosing surface obeys a characterization in terms of nonzero, Hermitian and indefinite operators has already been proven in lemma \ref{lem:equivS}. Conversely, suppose the states in $\partial U(\rho)$ are characterized by Eq.\,\eqref{eq:defNEnclA}. Hence, for a nonzero, indefinite and Hermitian operator $X$ there exists a real number $\lambda>0$ with $\lambda|\mathrm{Tr}X|<1$ such that
\begin{equation}
 \varrho=\frac{1}{p_-}\left(p_+\rho-\mathrm{sgn}(\mathrm{Tr}X)\lambda X\right)\in\partial U(\rho)
\end{equation}
where $p_\pm=\tfrac{1}{2}(1\pm\lambda|\mathrm{Tr}X|)$. Now, consider the map 
\begin{equation}\label{eq:Theta}
 \Theta_{\rho}(X)\equiv\mathrm{sgn}(\mathrm{Tr}X)\bigl[(\mathrm{Tr}X)\rho-X\bigr]
\end{equation}
defined on the set of nonzero, Hermitian and indefinite operators. The operator $Y=\Theta_\rho(X)$ represents a traceless, Hermitian operator and $Y=0$ if and only if $(\mathrm{Tr}X)\rho=X$ contradicting that $X$ is indefinite. Thus, we have $Y\neq0$ and one finds
\begin{align}
\varrho&=\frac{1}{p_-}\left\{p_+\rho-\lambda\bigl[|\mathrm{Tr}X|\rho-Y\bigr]\right\}\nonumber\\
&=\frac{p_+-\lambda|\mathrm{Tr}X|}{p_-}\rho+\frac{\lambda}{p_-}Y\nonumber\\
&=\rho+2\mu Y
\end{align}
where $\mu\equiv\lambda/(2p_-)>0$. It remains to show that $\Theta_\rho$ defines a surjection on the set of nonzero, Hermitian and traceless operators. This is obviously true as any traceless, nonzero, Hermitian operator $Y$ is necessarily indefinite and we have $\Theta_\rho(Y)=Y$. Hence, the set $\partial U(\rho)$ is indeed an enclosing surface.
\end{proof}

\end{document}